\documentclass[11pt]{article}
\usepackage{graphicx}
\usepackage{geometry,url}
\usepackage{amssymb,amsmath,amsthm}
\def\ket#1{\mathinner{|{#1}\rangle}}

{\catcode`\|=\active 
  \gdef\Braket#1{\left<\mathcode`\|"8000\let|\BraVert {#1}\right>}}
\def\BraVert{\egroup\,\mid@vertical\,\bgroup}

{\catcode`\|=\active
  \gdef\set#1{\mathinner{\lbrace\,{\mathcode`\|"8000\let|\midvert #1}\,\rbrace}}
  \gdef\Set#1{\left\{\:{\mathcode`\|"8000\let|\SetVert #1}\:\right\}}}
\def\midvert{\egroup\mid\bgroup}
\def\SetVert{\egroup\;\mid@vertical\;\bgroup}

\newtheorem{theorem}{Theorem} 
\newtheorem{proposition}[theorem]{Proposition} 

\newcommand{\oort}{\frac{1}{\sqrt{2}}}

\newcommand{\twoVec}[2]{\begin{pmatrix}#1\\#2\end{pmatrix}}
\newcommand{\IR}{{\mathbb{R}}}
\newcommand{\IC}{{\mathbb{C}}}

\begin{document}

\title{The Deutsch-Jozsa Problem: De-quantisation and Entanglement}
\author{Alastair A. Abbott\\ Department of Computer Science, University of Auckland,\\ 38 Princes St, Auckland 1010, New Zealand\\\texttt{aabb009@aucklanduni.ac.nz}}  
\maketitle

\begin{abstract}
The Deustch-Jozsa problem is one of the most basic ways to demonstrate the power of quantum computation. Consider a Boolean function $f : \{0,1\}^n \to \{0,1\}$ and suppose we have a black-box to compute $f$. The Deutsch-Jozsa problem is to determine if $f$ is constant (i.e. $f(x) = \text{const, } \forall x \in \{0,1\}^n$) or if $f$ is balanced (i.e. $f(x) = 0$ for exactly half the possible input strings $x \in \{0,1\}^n$) using as few calls to the black-box computing $f$ as is possible, assuming $f$ is guaranteed to be constant or balanced. Classically it appears that this requires at least $2^{n-1}+1$ black-box calls in the worst case, but the well known quantum solution solves the problem with probability one in exactly one black-box call. It has been found that in some cases the algorithm can be de-quantised into an equivalent classical, deterministic solution. We explore the ability to extend this de-quantisation to further cases, and examine with more detail when de-quantisation is possible, both with respect to the Deutsch-Jozsa problem, as well as in more general black-box algorithms.
\end{abstract}

\section{Introduction}
Deutsch's problem and the more general Deutsch-Jozsa problem were some of the first problems tackled in the field of quantum computing. They are simple, but are sufficiently non-trivial to be of interest. The generally accepted quantum solutions contain aspects of quantum parallelism, interference and entanglement, which are commonly cited as the main tools which give quantum computing its power. In this paper, we build on a previous result~\cite{Calude:2007p542} that shows Deutsch's original problem is solvable classically. We extend this note to the general Deutsch-Jozsa problem in an attempt to explore what is fundamentally important to quantum computation.

\subsection{Preliminaries}
In order to be able to talk about the differences between classical and quantum algorithms, we need define them in a way which captures their differences in a constructive manner. A \textit{classical algorithm} is one which is computed by a (probabilistic) Turing Machine or equivalent computational model, while a \textit{quantum algorithm} is one which is computed by a uniformly generated sequence of circuits $G=G_{T(n)}\dots G_1$, where each gate $G_i$ is either a unitary gate (chosen from a fixed set), or a measurement gate~\cite{Abbott:2010fk}.

Most quantum algorithms have a trivial classical counterpart via the simulation of their matrix-mechanical formulation. As long as one is careful,\footnote{If incomputable coefficients were used in the algorithm this simulation would clearly run into trouble, although it is not clear if such situations will arise in quantum computation. Alternatively, it has been shown quantum value indefiniteness allows quantum sources to generate incomputable sequences of bits~\cite{Calude:2008aa}; quantum computers can be used as such a source and trying to simulate this behaviour will clearly fail on a classical computer.} equivalent classical algorithms can be readily obtained by these means~\cite{Ekert:1998aa}. However, the dimension of Hilbert space (and thus the representation of the state) grows exponentially with the number of qubits used in a quantum algorithm, so a classical counterpart obtained by the trivial means takes space and time that is exponentially larger than the quantum algorithm does, leaving such simulations inefficient. In this paper we examine the existence of more efficient classical counterparts to quantum algorithms that are not exponential in time or space compared to the quantum algorithm. Some previous work in this has been done by Jozsa and Linden, who looked at conditions for simulating the matrix formulation of an algorithm efficiently~\cite{Jozsa:2003p1766}. We will look at the same issue from a slightly different angle with respect to the Deutsch-Jozsa problem, as well as in a more general black-box situation.

Throughout this paper we will use the shorthand notion $\ket{+}$ and $\ket{-}$ to represent the symmetric and antisymmetric equal superpositions of the basis states, so that the Hadamard gate $H$ has the following effect: $$H\ket{0} = \frac{1}{\sqrt{2}}(\ket{0}+\ket{1})=\ket{+}, \quad H\ket{1}=\frac{1}{\sqrt{2}}(\ket{0} - \ket{1}) = \ket{-}.$$

\subsection{Oracle Quantum Computations and Embeddings}

The main subject of this paper, the Deutsch-Jozsa problem (and also Deutsch's problem) is a form of an oracle computational problem~\cite{Berthiaume:1994la,Kitaev:2002it}. This means that the input is given to us as a black-box and the goal is to determine something about this black-box. It is important that this information can only be obtained by asking the oracle allowable questions. It must not be the case that examination of the black-box structure alone allows insight into the nature of the black-box~\cite[p. 554]{Deutsch:1992p560}.

In the current literature it is usually implicitly assumed that solving a problem with a classical black-box is equivalent to solving it with a quantum, or alternative kind of black-box. In general, a standard classical black-box can operate on classical bits only (0 or 1), while a quantum black-box can operate on any state in two-dimensional Hilbert space ($\mathcal{H}_2$). This difference in some sense appears to add extra complexity and power to the quantum black-box, and it is not clear that solutions can easily be compared without taking this into account. Some questions relating to this issue are developed more fully in~\cite{Abbott:2010fk}, and ongoing joint work with the University of York involves studying this in a physical Nuclear Magnetic Resonance implementation to gain further insight. In this paper we will adopt the standard stance and work with the natural generalisation of the classical black-box.

\section{Deutsch's Problem}

The original problem proposed by Deutsch~\cite{Deutsch:1985p605} is formulated as follows: consider a Boolean function $f : \{ 0,1\} \to \{0,1\}$, and suppose we are given a black-box (oracle) to compute $f$. Deutsch's problem is to determine if $f$ is constant (i.e. $f(0) = f(1)$) or balanced (i.e. $f(0) \neq f(1)$) in as few as possible calls to the black-box computing $f$.

\subsection{Quantum Solution}
A standard quantum solution for Deutsch's problem is briefly presented, as all further analysis will stem from this. This is based on the formulation given in~\cite{Cleve:1997p569} which solves Deutsch's problem with probability one using only one call to the quantum black-box computing $f$. A traditional classical algorithm would require two calls to a classical black-box in order to determine if $f$ is constant or balanced. The quantum black-box extends the classical black-box to operate on superpositions of basis states. Since the classical black-box is not reversible it is embedded in the unitary quantum black-box described by the unitary operator $U_f$, an $f$-controlled-NOT ($f$-cNOT) gate, such that 
$$U_f \ket{x}\ket{y} = \ket{x}\ket{ y \oplus f(x)}.$$
In order to see how the quantum solution works, note that
\begin{align*}
	U_f \ket{x}\ket{-} &= \ket{x}\oort(\ket{0\oplus f(x)}-\ket{1\oplus f(x)}) = (-1)^{f(x)}\ket{x}\ket{-}.
\end{align*}
From this observation we can formulate the quantum solution. Taking the initial state $\ket{0}\ket{1}$ and operating on it with a 2-qubit Hadamard gate $H^{\otimes 2}$, we get:
$$H^{\otimes 2}\ket{0} \ket{1} = H\ket{0} H\ket{1} = \ket{+}\ket{-}.$$
Next, operating on the state with $U_f$:
\begin{align}
U_f \ket{+} \ket{-} &= \frac{1}{\sqrt{2}} \left( (-1)^{f(0)}\ket{0} + (-1)^{f(1)} \ket{1} \right)\ket{-}\notag\\
	&= \frac{(-1)^{f(0)}}{\sqrt{2}}\left( \ket{0} + (-1)^{f(0)\oplus f(1)}\ket{1} \right)\ket{-},\label{eqn:Uf-n1}
\end{align}
and applying $H^{\otimes 2}$ one more time we get
$$H^{\otimes 2}\frac{(-1)^{f(0)}}{\sqrt{2}}\left( \ket{0} + (-1)^{f(0)\oplus f(1)}\ket{1} \right)\ket{-} = (-1)^{f(0)}\ket{f(0) \oplus f(1)} \ket{1}.$$
Measuring the first qubit we obtain $0$ with probability one if $f$ is constant and $1$ with probability one if $f$ is balanced.

This quantum solution is correct with probability one using only one call to the quantum black-box represented by $U_f$. An important note is that this computation involves no entanglement, and the auxiliary qubit remains unchanged by $U_f$ leaving the two qubits separable.\footnote{In this paper we take the terms `separable', `product-state' and `unentangled' to have the same meaning since we deal only with pure states. As such, we do attempt to generalise the results to mixed-state computations; see~\cite{Jozsa:2003p1766} for a discussion of the difficulties in extending general de-quantisation results to mixed-states.} The power of this quantum solution seems to come only from parallelism and interference, not entanglement.

\subsection{Classical Solutions}

While the quantum algorithm makes use of quantum parallelism and interference, these qualities (unlike entanglement) are not inherently quantum mechanical; their physical presence in the computation may be of quantum mechanical origin, but the mathematical effect and ability to be present in classical systems is due to the two-dimensionality of qubits compared to the one-dimensionality of classical bits. Hence, the same computational advantage should be achievable in a classical two-dimensional computational system.

The first method, presented in~\cite{Calude:2007p542}, uses complex numbers as classical two-dimensional bits.
The set $\{1,i\}$ acts as a computational basis in the same way that $\{\ket{0},\ket{1}\}$ does for quantum calculations, and its elements can similarly be assigned the meanings of `0' and `1'. Since an arbitrary complex number may be written as $z=a+bi$ with $a,b\in\IR$, a complex number $z$ is a natural superposition of the basis elements in the same way that a qubit is. 

Just as in the quantum case, we need to embed the original classical black-box computing $f$ in one which operates on complex numbers. This black-box can be represented by an operator $C_f$, a direct analogue of $U_f$ (although the requirement of unitarity is no longer necessary). The effect of $C_f$ (cf. Equation~\ref{eqn:Uf-n1}) is
$$C_f(a+bi)=(-1)^{f(0)}\left( a+(-1)^{f(0) \oplus f(1)}bi \right).$$
If $f$ is constant, $C_f$ is the identity operation to within a factor of $-1$ ($C_f(x) = \pm x$). If $f$ is balanced, $C_f$ is the conjugation operation ($C_f(x) =\pm \overline{x}$). In order to project our complex numbers back on to the computational basis, we multiply the black-box output by the input. Note that this is not a physical operation but an abstract mathematical one; in the quantum case we are forced to work in a physical computational embedding, but classically this is not the case.

If $z=1+i$ (an equal superposition of basis states),
$$\frac{1}{2}z \times C_f(z) = 
\begin{cases}
	\frac{\pm 1}{2}z^2 =\pm i & \text{if $f$ is constant,}\\
	\frac{\pm 1}{2}z \overline{z} = \pm 1 & \text{if $f$ is balanced.}\\
\end{cases}$$ In this manner, if the output is imaginary then $f$ is constant, if it is real then $f$ is balanced. Importantly, this is a deterministic result, and in fact the sign of the output allows us to identify \emph{which} balanced or constant function $f$ is. This is something the quantum algorithm provably cannot do~\cite{Mermin:2007aa}.

This process of finding a classical solution with the same complexity as the quantum solution for a problem informally corresponds to what we call de-quantisation. Formally, the requirements for a classical algorithm to be a de-quantisation are different depending on the type of problem being solved. Let $\mathcal{A}$ be quantum algorithm with complexity measure $f(n)$ and output probability distribution $\mathcal{P}$. A probabilistic Turing machine $M$, such that for every computable $\gamma > 0$ there effectively exists a probability distribution $\mathcal{P'}$  with $|\mathcal{P} - \mathcal{P'}| < \gamma$, and $M$ sampling from $\mathcal{P'}$ has complexity measure $g(n,\gamma)$, is a potential de-quantisation of $\mathcal{A}$~\cite{Abbott:2010fk}. For $M$ to be a de-quantisation, $g$ must satisfy certain requirements depending on the type of problem. For example, in a standard (non black-box) algorithm, $f(n) = T(n)$, the size of the circuit for $\mathcal{A}$, and $M$ is a de-quantisation if $g(n,\gamma) = \text{poly}(T(n),\log{(1/\gamma)})$. In the Deutsch-Jozsa algorithm $f(n)=1$ is the number of black-box calls, and $M$ is a de-quantisation if $g(n,\gamma) = f(n)$, i.e. the classical solution must also use only one black-box call. For an overview and comparison of various de-quantisation techniques we refer the reader to~\cite{Abbott:2010fk}, but we focus on entanglement based techniques in this paper.

The de-quantisation we have presented is only one possible method, and places emphasis on mathematical correspondence with the quantum solution. A different solution is presented by Arvind in~\cite{Arvind:2001p561} which draws physical similarities that are more visible than in the above solution. Arvind uses the polarisation of a photon (treated classically) as the computational basis $\{x\text{-pol}$, $y\text{-pol}\}$, and any polarisation in the $xy$-plane is physically valid. It is noted that all transformations in the group SU(2) can be realised by two quarter-wave plates and a single half-wave plate orientated suitably, and Deutsch's problem is solved using such transformations. Written in matrix form the solution is mathematically identical to the quantum one. This corresponds to the following physical process: preparing a photon with $y$-polarisation, rotating its polarisation anti-clockwise in the $xy$-plane by $45^{\circ}$, applying the optical black-box and applying the anti-clockwise rotation once more before measuring the $y$-polarisation of the photon.
 
The correspondence here relies not on embedding classical bits in a different, classical two-dimensional basis but on directly implementing the transformations used in the quantum solution through classical means. In other words, the quantum algorithm does not take advantage of non-classical effects, so the same result can be obtained through purely classical optics (regardless of whether or not photons are actually quantum mechanical objects---they do not need to be treated as such).

\section{The Deutsch-Jozsa Problem}

This problem was extended by Deutsch and Jozsa~\cite{Deutsch:1992p560} to functions on $n$-bit strings. The standard formulation of the Deutsch-Jozsa problem is as follows: let $f : \{0,1\}^n \to \{0,1\}$, and suppose we are given a black-box computing $f$ with the guarantee that $f$ is either constant (i.e. $\forall x \in \{0,1\}^n: f(x) = a$, $a \in \{0,1\}$) or balanced (i.e. $f(x) = 0$ for exactly half of the possible inputs $x \in \{0,1\}^n$). Such a function $f$ is called \emph{valid}. The Deutsch-Jozsa problem is to determine if $f$ is constant or balanced in as few black-box calls as possible. An important note is that unlike in Deutsch's problem, where there are exactly two balanced and two constant functions $f$, the distribution of constant and balanced functions is asymmetrical in the Deutsch-Jozsa problem. In general, there are $N=2^n$ possible input strings, each with two possible outputs ($0$ or $1$). Hence, for any given $n$ there are $2^N$ possible functions $f$. Of these, exactly two are constant and $\binom{N}{N/2}$ are balanced. Evidently, the probability that our $f$ is constant tends towards zero very quickly (recall $f$ is guaranteed to be valid). Furthermore, the probability that any randomly chosen function of the $2^N$ possible functions is valid is:
$$\frac{\binom{N}{N/2}+2}{2^N},$$
which also tends to zero as $n$ increases. This is clearly not an ideal problem to work with, however this does not mean that we cannot gain useful information from studying it.

\section{Adaptation for $n=2$}
In this section we will provide a formulation of the solution for the Deutsch-Jozsa problem with $n=2$ which makes the separability of the states clearly evident and draws obvious parallelism with the $n=1$ solution presented in the previous section.

\subsection{Quantum Solution}

For $n=2$ the quantum black-box we are given takes as input three qubits and is represented by the following unitary operator $U_f$, just as it was for $n=1$:
$$U_f \ket{x} \ket{y} = \ket{x} \ket{y \oplus f(x)},$$
where $x \in \{0,1\}^2$. For $n=2$ there are sixteen possible Boolean functions. Two of these are constant, another six are balanced and the remaining eight are not valid. All these possible functions are listed in Table~\ref{table:bFuncs}.

\begin{table}[ht]
\begin{center}
\begin{tabular}{c|cc|cccccc|cccccccc}
$f(x)$ & \multicolumn{2}{c|}{Constant} & \multicolumn{6}{c|}{Balanced} & \multicolumn{8}{c}{Invalid}\\
\hline
$f(00) = $ & \phantom{.} 0 \phantom{.} & 1 & 0 & 1 & 0 & 1 & 0 & 1 & 1 & 0 & 1 & 0 & 1 & 0 & 0 & 1 \\
$f(01) = $ & 0 & 1 & 0 & 1 & 1 & 0 & 1 & 0 & 1 & 0 & 1 & 0 & 0 & 1 & 1 & 0 \\
$f(10) = $ & 0 & 1 & 1 & 0 & 1 & 0 & 0 & 1 & 1 & 0 & 0 & 1 & 1 & 0 & 1 & 0 \\
$f(11) = $ & 0 & 1 & 1 & 0 & 0 & 1 & 1 & 0 & 0 & 1 & 1 & 0 & 1 & 0 & 1 & 0 \\
\end{tabular}
\end{center}
\caption{All possible Boolean functions $f : \{0,1\}^2 \to \{0,1\}$.}
\label{table:bFuncs}
\end{table}
Evidently, half of these functions are simply the negation of another. If we let $f'(x) = f(x) \oplus 1$, we have:
\begin{align*}
U_{f'}\ket{x}\ket{-} &= (-1)^{f'(x)}\ket{x}\ket{-}
	 = - U_f \ket{x}\ket{-}.
\end{align*}
In this case the result obtains a global phase factor of $-1$; this has no physical significance so the outputs of $U_f$ and $U_{f'}$ are indistinguishable.

If we initially prepare our system in the state $\ket{00}\ket{1}$, operating on this state with $H^{\otimes 3}$ gives
\begin{equation}
\label{eqn:equalSuperpos}
H^{\otimes 3}\ket{00}\ket{1} = \frac{1}{2}\sum_{x\in \{0,1\}^2}\ket{x}\ket{-} = \ket{++}\ket{-}.
\end{equation}
After applying the $f$-cNOT gate $U_f$ we have
\begin{equation}
\label{eqn:n2fcNot}
\begin{split}
U_f \frac{1}{2}\sum_{x\in \{0,1\}^2}\ket{x}\ket{-} & = \sum_{x\in \{0,1\}^2}\frac{(-1)^{f(x)}}{2}\ket{x}\ket{-}.
\end{split}
\end{equation}
From the well known rule (see~\cite{Jorrand:2003p562}) about 2-qubit separable states, we know that this state is separable if and only if $(-1)^{f(00)}(-1)^{f(11)} = (-1)^{f(01)}(-1)^{f(10)}$. By noting that the mapping $(-1)^{f(a)}(-1)^{f(b)} \leftrightarrow f(a) \oplus f(b)$ is a bijection, we see that the separability condition reduces to $f(00) \oplus f(11) = f(01) \oplus f(10)$. From Table~\ref{table:bFuncs} it is clear this condition must hold for all balanced or constant functions $f$ for $n=2$, and thus no entanglement is present in the $n=2$ case.

We can now rewrite Equation~\ref{eqn:n2fcNot} in a separable form as follows:
\begin{equation}
\label{eqn:n2Uf}
U_f \ket{++} \ket{-} = \frac{(-1)^{f(00)}}{2} \left( \ket{0} + (-1)^{f(00) \oplus f(10)} \ket{1} \right) \left( \ket{0} + (-1)^{f(10) \oplus f(11)} \ket{1} \right) \ket{-}.
\end{equation}

%
By applying a final 3-qubit Hadamard gate to project this state onto the computational basis we obtain
\begin{multline*}
$$\frac{(-1)^{f(00)}}{2} H^{\otimes 3} \left( \ket{0} + (-1)^{f(00) \oplus f(10)} \ket{1} \right) \left( \ket{0} + (-1)^{f(10) \oplus f(11)} \ket{1} \right) \ket{-} \\= (-1)^{f(00)} \ket{f(00) \oplus f(10)}  \otimes \ket{f(10) \oplus f(11)} \ket{1}.$$
\end{multline*}
By measuring both the first and second qubits we can determine the nature of $f$: if both qubits are measured as $0$, then $f$ is constant, otherwise $f$ is balanced.

\subsection{Classical Solutions}

Because the quantum solution contains no entanglement, the problem can be de-quantised in a similar way to the $n=1$ case, but this time using two complex numbers as the input to the black-box. We extend the black-box to operate on two complex numbers, $C_f : \IC^2 \to \IC^2$, and define it by analogy to $U_f$ just as we did for the $n=1$ case. Let $z_1$, $z_2 $ be complex numbers,
\begin{equation}
	\label{deQuant:eqn:2bitCf}
	C_f\twoVec{z_1}{z_2} = C_f\twoVec{a_1+b_1i}{ a_2 + b_2i} =  \twoVec{(-1)^{f(00)}\left[a_1 + (-1)^{f(00) \oplus f(10)}b_1i\right]}{ a_2 + (-1)^{f(10) \oplus f(11)}b_2i}.
\end{equation}
\if02
It is important to note that, just as in the quantum case where the output of the black-box was two qubits that can be independently measured, the output of $C_f$ is two complex numbers  that can be independently manipulated, rather than the complex number resulting from their product. This is fairly intuitive because the ability to measure specific bits (of any kind) is fundamental to computation. Note, however, that in a quantum system measurement of an entangled qubit necessarily effects the state of the other, unmeasured, qubits it is entangled to.
\fi
If we let $z = z_1 =  z_2 = 1+i$, multiplying by $z$ to project onto the computational basis as for the $n=1$ case, we get:
\begin{displaymath}
	\frac{z}{2}\times C_f \twoVec{z}{z} = \frac{1}{2}\times
	\begin{cases}
		\twoVec{(-1)^{f(00)}z^2}{z^2} = \twoVec{(-1)^{f(00)}i}{i} & \text{if $f$ is constant,}\\
		&\\
		\twoVec{(-1)^{f(00)}z\bar{z}}{z^2} = \twoVec{(-1)^{f(00)}}{i} &\\
		\twoVec{(-1)^{f(00)}z\bar{z}}{z\bar{z}} = \twoVec{(-1)^{f(00)}}{1} & \text{if $f$ is balanced.}\\
		\twoVec{(-1)^{f(00)}z^2}{z\bar{z}} = \twoVec{(-1)^{f(00)}i}{1} &\\
	\end{cases}
\end{displaymath}

By checking both of the resulting complex numbers, we can determine whether $f$ is balanced or constant with certainty. If both complex numbers are imaginary then $f$ is constant, otherwise it is balanced. In fact, just as in the $n=1$ case, the ability to determine if the output numbers are negative or positive allows us to determine the value of $f(00)$ and thus which Boolean function $f$ is.

It is because of the separability of the quantum solution that it is possible to reproduce the ability to solve the problem with only one black-box call. If the quantum state was not separable there would be no embedding of the original black-box into one operating on two complex-bits which would allow us to determine the nature of $f$ with only one black-box call.

As with the $n=1$ case, an alternative classical approach can be presented using two photons. If a transformation on two qubits can be separated into a transformation on each qubit independently then the transformation is trivially implemented classically. As Arvind noted~\cite{Arvind:2001p561}, the 2-qubit transformations $U_f$ can be written in such a way since they are separable: from Equation~\ref{eqn:n2Uf} it is clear $U_f$ can be written as a product of the two 1-qubit gates\footnote{So far we have been considering the case where $U_f$ operates on $n$ input qubits and one auxiliary qubit, $\ket{-}$. It has been shown (see~\cite{Collins:1998p1153}) that the auxiliary qubit is not necessary if we restrict ourselves to the subspace spanned by $\ket{-}$. We have presented the algorithm with the auxiliary qubit present because it is more intuitive to think of the input-dependent phase factor being an eigenvalue of the auxiliary qubit which is kicked back. The de-quantised solutions, however, bear more resemblance to this reduced version of $U_f$ operating only on $n$ qubits.}
\begin{align*}
U_f^{(1)} \ket{+} = \frac{1}{\sqrt{2}}\left(\ket{0} + (-1)^{f(00)\oplus f(10)}\ket{1}\right),\\
U_f^{(2)} \ket{+} = \frac{1}{\sqrt{2}}\left(\ket{0} + (-1)^{f(10)\oplus f(11)}\ket{1}\right).
\end{align*}
Each of these are valid unitary operators, and the transformation describing the black-box may be written $U_f = U_f^{(1)}\otimes U_f^{(2)}$. This means that the operation of $f$ can be computed by applying a 1-qubit operation (implemented as wave-plates) to each photon independently, and thus a classical solution for $n=2$ is easily found. The photons need not interact with each other at any point during the algorithm, not even inside the black-box implementation.

This classical, optical method is equivalent to both the quantum solution and the previously described classical solution (although the latter is abstract and mathematical, while the former are physical). The difference is in how it is represented, bringing emphasis on the fact that for $n=2$ the quantum solution does not take advantage of uniquely quantum behaviour and is thus classical in nature. Further, it shows that the solution can be obtained without any interaction or sharing of information between qubits.

\section{Separability in the Deutsch-Jozsa problem}
\label{sec:deq:sep}

Before we consider de-quantisation of the Deutsch-Jozsa problem for $n \ge 3$, we will first examine the separability of the states used in the quantum solution more carefully, as determining if a state is separable is a key step in determining if an easy de-quantisation is possible. Conditions to determine if a $n$-qubit state is separable are presented in~\cite{Jorrand:2003p562}. We will review these results, before reformulating them in a recursive manner which will allow us to apply these results much more easily to the Deutsch-Jozsa problem.

In order to determine if an arbitrary $n$-qubit state is separable, we must first introduce the concept of pair product invariance~\cite{Jorrand:2003p562}. A state $\ket{\psi_n} = \sum_{i=0}^{N-1}\alpha_i \ket{i}$ with $N=2^n$ is \emph{pair product invariant} if and only if $\forall k \in \{1,\dots,n\}, \forall i \in \{0,\dots,K-1\} : \alpha_i \alpha_{K-i-1} = c_k$, where each $c_k$  is a constant and $K=2^k$.

Pair product invariance can also be reformulated recursively. We let $P_n$ be the set of all pairs $(i,k)$ such that for any two elements $(a,k),(b,k)\in P_n$ we have $\alpha_a\alpha_{K-a-1} = \alpha_b\alpha_{K-b-1} = c_k$ if and only if $\ket{\psi}$ is pair product invariant. We can recursively write this by breaking up the iteration over all $k \in \{1,\dots,n\}$. To do so we define $P_l$ as:
\begin{equation}
	\label{P}
	P_l = \{ (i,k)\mid k\in\{2,\dots,l\}, i\in\{0,\dots,2^{k-1}-1\} \}.
\end{equation}
Recursively, this becomes:
\begin{align*}
	P_{l} &= P_{l-1} \cup \{ (i,l)\mid i \in \{0,\dots,2^{l-1}-1\}  \}\text{, with the base case}\\
	P_2 &= \{ (0,2), (1,2) \}.
\end{align*}
\begin{theorem}
	\label{thm:PPIRec}
	An $n$-qubit state $\ket{\psi_n} = \sum_{i=0}^{N-1}\alpha_i \ket{i}$ is pair product invariant if and only if $P_n$, as defined by \eqref{P}, satisfies the following condition: $\forall (a,k),(b,k) \in P_n : \alpha_a\alpha_{K-a-1} = \alpha_b\alpha_{K-b-1} = c_k$ where $K=2^k$ and $c_k$ is constant.
\end{theorem}
\begin{proof}
Firstly, note that the base case, $P_2$, simply amounts to the well-known condition that $\alpha_{00}\alpha_{11}=\alpha_{01}\alpha_{10}$. We can modify the definition of pair product invariance by changing $i$ to range from $0$ to $2^{k-1} - 1$ instead of $K-1$ since $\alpha_i \alpha_{K-i-1} = \alpha_{K-i-1} \alpha_i$, thus avoiding double counting. The $k=1$ case can also be removed since it now reduces to a single term, and is hence unnecessary. The definition of $P_l$ in \eqref{P} ensures that the quantifier $\forall (a,k),(b,k) \in P_n$ runs for $k\in \{2,\dots,n\}$, $a,b\in \{0,\dots,2^{k-1}-1\}$. It is then clear to see that the required condition is satisfied by $P_n$ if and only if $\ket{\psi}$ is pair product invariant.\qed
\end{proof}

The main theorem of relevance to us regarding pair product invariance is Theorem~\ref{deq:thm:PPI} below. 
\begin{theorem}[Theorem 1 in~\cite{Jorrand:2003p562}]
	\label{deq:thm:PPI}
 	If $\forall i \in \{0,\dots,N-1\} : \alpha_i \neq 0$, a state $\ket{\psi_n}=\sum_{i=0}^{N-1}\alpha_i\ket{i}$ is fully separable if and only if it is pair product invariant.
\end{theorem}
Hence, in order to determine if a state $\ket{\psi_n}$ is separable we can simply apply Theorems~\ref{thm:PPIRec} and~\ref{deq:thm:PPI}. The recursive formulation of pair product invariance will prove more useful to develop general results on separability.

Note that any constant function produces a pair product invariant state, as all $\alpha_i$ are equal and the conditions are trivially satisfied. The output of the quantum black-box for constant $f$ can be separated as:
\begin{align*}
	\frac{1}{2^{n/2}}U_f \sum_{x \in \{0,1\}^n}\ket{x}\ket{-} & = \frac{1}{2^{n/2}}\sum_{x \in \{0,1\}^n}(-1)^{f(x)}\ket{x}\ket{-}\\
	&= \frac{(-1)^{f(00)}}{2^{n/2}}\sum_{x \in \{0,1\}^n}\ket{x}\ket{-}\\
	&= (-1)^{f(00)} \ket{+}^{\otimes n} \ket{-}.
\end{align*}

\subsection{The case of $n \ge 3$}

For $n=3$ we are able to find balanced functions which do not produce pair product invariant states, and thus, by Theorem~\ref{deq:thm:PPI}, these states are entangled.

If we choose $f$ such that $$\left( f(0), f(1), f(2), f(3), f(4), f(5), f(6), f(7) \right) = (0,0,0,1,1,1,1,0),$$ then this $f$ is obviously balanced. For this choice of $f$, $f(000) \oplus f(011) \neq f(001) \oplus f(010)$. This implies that $\alpha_{000}\alpha_{011} \neq \alpha_{001}\alpha_{010}$ and hence the output of $U_f$ in the standard quantum algorithm is not pair product invariant and is thus entangled.

The recursive formulation of pair product invariance allows us to determine exactly how many separable states exist for any given $n$. Theorem~\ref{thm:PPIRec}, along with the fact that $P_n \subset P_{n+1}$, means that for a function $f_{n+1}$ to produce a separable state, the first $n$ qubits of the resulting state must be pair product invariant. Hence we see that all functions $f_{n+1}$ which produce separable states are based on another such function $f_n$, and hence satisfy $f_{n+1}(0\circ x_n) = f_n(x_n), \forall x_n\in \{0,1\}^n$, where $0\circ x_n$ denotes the concatenation of $x_n$ onto $0$. This determines the action of $f_{n+1}$ on half of the possible input strings, and we must determine the possible actions on the other half. Figure~\ref{fig:deq:ppi} shows this nature of pair product invariance.
\begin{figure}[ht]
\begin{centering}
	\includegraphics[scale=0.85]{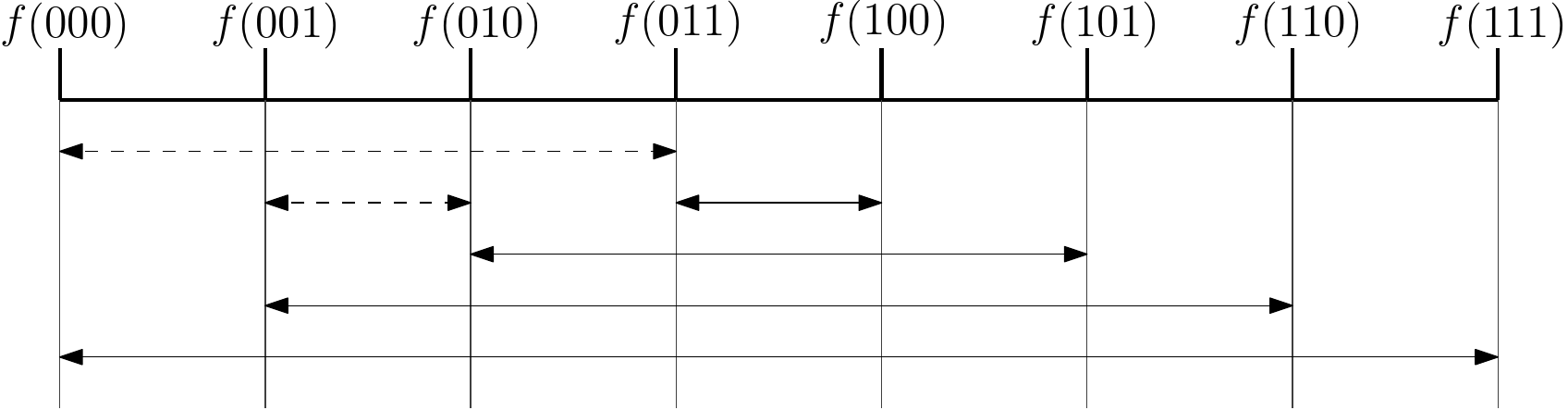}
	\caption[An illustration of the recursive nature of pair product invariance.]{A diagram showing the recursive nature of pair product invariance. All pairs of states linked by the same type of arrow must have the same parity under addition modulo $2$.}
	\label{fig:deq:ppi}
\end{centering}
\end{figure}

The action of $f_{n+1}$ on the remaining $N=2^n$ input strings is determined by Theorem~\ref{thm:PPIRec}. This requires that $\alpha_{N-1}\alpha_{N} = \alpha_{N-2}\alpha_{N+1} = \cdots = \alpha_{0}\alpha_{2N-1}$. Since $\alpha_{N-1}$ is already determined, $\alpha_N$ can take on two possible values. However, once this value is chosen, all $\alpha_i$ for $i>N$ are uniquely determined by the Theorem~\ref{thm:PPIRec}. If we let $a_n$ be the number of functions $f_n$ which produce separable states, then $f_{n+1}$ can have $a_n$ possible configurations for acting on $f_{n+1}(0\circ x_n)$, and for each of these configurations, there are two configurations for the remaining $N$ input strings, hence $a_{n+1} = 2a_n$.
Using the fact that $a_1 = 4$ we get an explicit result for the number of Boolean functions $f_n$ such that $U_{f_n}\ket{x}\ket{-}$ is separable:
$$a_n = 2^{n+1}.$$

We see that the number of separable states increases exponentially with the number of qubits being used. We also know that the number of possible functions $f_n$ which are either balanced or constant is
$$b_n = \binom{N}{N/2} + 2 = \binom{2^n}{2^{n-1}} + 2.$$
The fraction of possible Boolean functions which can be separated is 
$$\frac{a_n}{b_n}=2^{n+1}/\left(\binom{2^n}{2^{n-1}}+2\right).$$
This tends towards zero extremely quickly even for small $n$. Further, we can verify that for all $n\ge 3$ there exist valid functions $f$ such that the output of the black-box is \emph{maximally} entangled, i.e. all $n$ qubits are entangled together.
\begin{proposition}
	For all $n\ge 3$ there exists a valid $f : \{0,1\}^n \to \{0,1\}$ such that the first $n$ qubits of $U_f \ket{+}^{\otimes n}\ket{-} = \ket{\psi}\ket{-}$ are all entangled.
\end{proposition}
\begin{proof}
	Firstly we note that if any qubit is separable from the rest, a specific sub-structure is present in the state-vector. Without loss of generality assume the $i$th qubit, $\ket{\psi_i}$, is separable from the rest, and let $\ket{\psi'}=\ket{\psi_1}\cdots \ket{\psi_{i-1}}\ket{\psi_{i+1}}\cdots \ket{\psi_n}$ be the state of the other $n-1$ qubits. Since all the amplitudes in $\ket{\psi}$ are $\pm 2^{-n/2}$, we know $\ket{\psi_i} = \ket{\pm} = (1,\pm 1)$---we use the state-vector notation with normalisation implicit. Let us write the state vector of $\ket{\psi'}$ in blocks of size $2^{n-i}$ as $\ket{\psi'}=(\mathbf{x}_1,\dots,\mathbf{x}_{2^{i-1}})$. Because of the symmetry in the state vector due to the tensor-product we have $\ket{\psi}=(\mathbf{x}_1,\pm \mathbf{x}_1,\dots,\mathbf{x}_{2^{i-1}},\pm\mathbf{x}_{2^{i-1}})$ where all the signs $\pm 1$ are identical. The function defined by 
	$$f(a) = 
	\begin{cases}
		0 & \text{for $a=0,1,\dots,2^{n-1}-2,2^{n}-2$,}\\
		1 & \text{elsewhere,}
	\end{cases}$$
	 is clearly valid, but the corresponding state $\ket{\psi}$ for this $f$ is 
	$$\ket{\psi} = (\underbrace{1,\dots,1,1}_{2^{n-1}-1 \text{ times}},-1,\underbrace{-1,\dots,-1}_{2^{n-1}-2 \text{ times}},1,-1),$$
which can easily be verified to not have the required sub-structure, and is thus maximally entangled.\qed
\end{proof}

The result of these observations is that even if we are promised $f$ is valid, we can no longer be sure the output of the black-box is separable, and for $n \ge 3$ the probability that it is separable tends to zero very quickly. Since any de-quantisation would have to handle maximally entangled outputs of the black-box, this method of de-quantisation used for $n=1,2$ will not scale easily to higher $n$, and in general yields very little information about the nature of a function. Looking at this from the view of computation with classical photons, there is no classical physical equivalent of entangled photons, as this is a purely quantum mechanical effect. However, in general it is very hard to show that no de-quantisation exists for a quantum algorithm which does not introduce exponential increase in space or time. In most cases, as earlier mentioned, a trivial method of de-quantisation is possible, but to show no better de-quantisation exists is very hard. As an example, Jozsa and Linden showed~\cite{Jozsa:2003p1766} that using the stabiliser description of quantum computation~\cite{Gottesman:1999aa} it is possible to find an efficient classical simulation of a quantum algorithm containing unbounded entanglement.

\section{General de-quantisation}

While de-quantisation appears to be very hard in the Deutsch-Jozsa problem for $n\ge 3$, we can investigate the general ability to de-quantise black-box algorithms. The main task in trying to de-quantise such an algorithm is to de-quantise the black-box. Indeed, efficiently simulating a black-box algorithm requires finding a classical black-box embedding which doesn't use an exponentially increasing amount of time or space. Using the methods looked at in this paper, this step would initially require showing that both the input and output of the black-box are separable. 

Firstly, we will summarise some important results due to Jozsa and Linden~\cite{Jozsa:2003p1766}, who examined the ability to efficiently simulate standard quantum algorithms by simulating the matrix mechanical formulation of these algorithms. We call a qubit register $\ket{\psi}$ $p$-blocked if at every step of the computation no subset of more than $p$ qubits are entangled. Their main result is Theorem~\ref{thm:sim}.
\begin{theorem}
	\label{thm:sim}
	Let $\mathcal{A}$ be a quantum algorithm with the properties that only the final gate is a measurement gate, and after every unitary gate application the state $\ket{\psi}$ is $p$-blocked, with $p$ independent of the input size $n$. Then the algorithm $\mathcal{A}$ can be de-quantised using the matrix formulation.
\end{theorem}
The proof relies on breaking up the $2^n \times 2^n$ matrices describing the gates with which the circuits in $\mathcal{A}$ are composed of into matrices operating on no more than $2p$ qubits at once.\footnote{The $2p$ comes about rather than $p$ in the case where a 2-qubit gate may operate on qubits from two separate entangled blocks. This approach assumes the circuit is decomposed into 1-qubit and 2-qubit gates.} Each matrix corresponding to a gate is replaced by a single matrix no larger than $2^{2p} \times 2^{2p}$. Since $p$ is constant, the cost of directly simulating the quantum algorithm by matrix multiplication can be reduced to a linear overhead, exponential in $p$ rather than $n$. 

This fixed-parameter tractability approach is mathematically equivalent to our approach, although our approach has the advantage that it retains more similarity to the quantum algorithms. Our method can also be easier to apply and produces simpler results than blind simulation of the matrix mechanics, and also more easily extends to the setting of black-box computation. The approach used by Jozsa and Linden does not directly apply to black-box computation where the nature of the black-box is unknown. It could readily be adapted by having the black-box perform a set of matrix multiplications, but such an approach is less natural. We will explore applying our de-quantisation method to arbitrary separable black-box computations since it is easier to examine and work with. Note that it is also not surprising that de-quantisation does not easily extend in the Deutsch-Jozsa problem since we have shown the entanglement grows exponentially.

The simplest case to tackle for de-quantising an arbitrary black-box algorithm is the case that both the input and output of the black-box $U_f$ can be expressed in a separable form. In this case we can show how a simple de-quantisation in the spirit presented previously can easily be obtained. A nice feature which makes de-quantisation of a black-box algorithm simple is that we do not have to worry about considering the separability of the decomposition of the gate, as it is supplied as an arbitrarily complex unitary gate, not necessarily decomposed into gates from a universal basis (indeed, we have no way of knowing how the black-box is devised). If we knew that the input and output of $U_f$ were separable, but had to decompose it into smaller unitary gates, we could not guarantee separability throughout the decomposed circuit representing $U_f$. However, the unitarity of $U_f$, regardless of its dimension, will allow a simple de-quantisation.
\begin{theorem}
	Let $U_f$ be the unitary operator representing the black-box such that $U_f$ never entangles its input in the quantum algorithm $\mathcal{A}$, i.e. both the input and output of $U_f$ are separable. Then $\mathcal{A}$ can be de-quantised into a classical algorithm with the same number of black-box calls.
\end{theorem}
\begin{proof}
	We can write the action of $U_f$ under these assumptions as
	\begin{equation}
		\label{eqn:deq:BBproof}
		U_f \ket{\psi_1}\ket{\psi_2}\cdots \ket{\psi_n} = \ket{\phi_1}\ket{\phi_2}\cdots \ket{\phi_n},
	\end{equation}
	and hence we can write $U_f = U_f^{(1)} \otimes U_f^{(2)} \otimes \cdots \otimes U_f^{(n)}$, where each $U_f^{(i)}$ acts on the $i$th qubit only. Since the quantum amplitudes can, in general, be complex-valued, we give a general construction of the classical black-box $C_f$ acting on complex-vectors rather than the complex-bits we used previously. This method amounts to simulating the matrix-mechanical formulation, much as in Theorem~\ref{thm:sim}. Let $$U_f^{(i)}=\begin{pmatrix}a & b\\c & d\end{pmatrix},$$ then define 
	$$C_f \begin{pmatrix}(\alpha_1,\beta_1)\\\vdots \\ (\alpha_n,\beta_n) \end{pmatrix} = \begin{pmatrix}(a\alpha_1 + b \beta_1,c\alpha_1 + d\beta_1)\\\vdots \\ (a\alpha_n + b \beta_n,c\alpha_n + d\beta_n) \end{pmatrix}.$$ A de-quantisation of $\mathcal{A}$ can then simulate the matrix evolution of the quantum algorithm, but replace queries to the quantum black-box by ones to the classical black-box represented by $C_f$ operating on $n$ two-component vectors. Thus the de-quantised algorithm will call the black-box exactly the same number of times as the quantum one.\qed
\end{proof}
\if01
\begin{theorem}
	Let $U_f$ be the unitary operator representing the black-box, and assume that for the set of input states used by the quantum algorithm $\mathcal{A}$ we wish to de-quantise, both the input and output $U_f$ are separable. Then the black-box can be de-quantised into an efficient classical solution.
\end{theorem}
\begin{proof}
	We can write the action of $U_f$ under these assumptions as
	\begin{equation}
		U_f \ket{\psi_1}\ket{\psi_2}\cdots \ket{\psi_n} = \ket{\phi_1}\ket{\phi_2}\cdots \ket{\phi_n}.
	\end{equation}
	Since $U_f$ is unitary, there must exist a unitary inverse $U_f^\dag$, which has the effect
	\begin{equation}
		\label{eqn:deq:BBproofDag}
		U_f^\dag \ket{\phi_1}\ket{\phi_2}\cdots \ket{\phi_n} = \ket{\psi_1}\ket{\psi_2}\cdots \ket{\psi_n}.
	\end{equation}
	From \eqref{eqn:deq:BBproof} we see that we each qubit undergoes the transformation $\ket{\psi_i} \to \ket{\phi_i}$, but we must show this transformation is unitary. Let us write $U_f = U_f^{(1)} \otimes U_f^{(2)} \otimes \cdots \otimes U_f^{(n)}$ Since the tensor product is distributive with respect to adjoints, we can write $U_f^\dag = (U_f^{(1)})^\dag \otimes (U_f^{(2)})^\dag \otimes \cdots \otimes (U_f^{(n)})^\dag$. Hence, we see that $U_f^{(i)}\ket{\psi_i} = \ket{\phi_i}$ and $(U_f^{(i)})^\dag \ket{\phi_i} = \ket{\psi_i}$. Thus the single qubit transformation is indeed unitary and we can write $U_f$ as the product of single-qubit unitary transformation above. It is then clear to see that de-quantisation of the black-box is possible. Evidently, the optical method we previously employed can trivially be used since each $U_f^{(i)} \in \text{SU}(2)$ and the black-box is then constructed out of wave plates. We only need to prepare the $n$ photons in the correct state, operate on them with the black-box and continue the rest of the algorithm with the transformed photons. Equivalently, a method using two-dimensional classical bits such as complex numbers could easily be used.\qed
\end{proof}\fi
This shows that, as expected, a quantum black-box algorithm can be de-quantised easily into an equivalent classical algorithm if the black-box never entangles the input. Note that while the de-quantised algorithm might be inefficient with respect to time, the number of black-box calls is the same as for the quantum algorithm. This black-box de-quantisation could be extended to allow bounded entanglement as in Theorem~\ref{thm:sim}. If entanglement is bounded throughout the whole algorithm, Theorem~\ref{thm:sim} will guarantee the de-quantisation is also efficient in time. These de-quantisation tools provide a useful technique of developing new classical algorithms from ones which are more naturally expressed in the quantum world. If unbounded entanglement is present, a successful de-quantisation will need to handle the problem using a different representation of the process.

\section{Conclusion}
We have examined the ability to de-quantise the Deutsch-Jozsa problem for various values of $n$ in order to gain a better understanding of quantum algorithms and the ability for them to give exponential improvements over classical algorithms. We have extended the de-quantisation presented in~\cite{Calude:2007p542} to the $n=2$ case by utilising separability of the quantum algorithm. We have shown that for $n>2$ there exist many balanced Boolean functions $f$ for which the output of $U_f$ is entangled. The fraction of balanced functions which are separable has been shown to approach zero very rapidly. This tells us that if we were to pick a random Boolean function which is constant or balanced, the probability of being able to learn information about the nature of the function through classical means in one black-box call tends to zero.

The systematic method of tackling quantum algorithms and searching for classical counterparts helps us to see where quantum algorithms get their power from. Trying to understand this is an extremely important step in the process of trying to devise new quantum algorithms. In order to make good quantum algorithms with ease, we need to have a much better understanding than we currently do about where their power comes from, and how to use this effectively. In our investigation we obtained conditions for which, if satisfied, indicate a black-box algorithm can be de-quantised. These kind of conditions, if explored further, are a step towards better measures of the usefulness of quantum algorithms.

An area that still needs to be looked into much further is that of the black-box complexity. It is not clear the standard procedure of ignoring the differences in complexity between quantum and classical black-boxes is entirely valid. Without a proper method to compare different types of black-box algorithms it is hard to convincingly claim one is better than the other, because they might not be solving quite the same problem. We see this as an important issue to investigate further.

\section*{Acknowledgment}
The author would like to thank Prof. Cristian Calude for motivating this research and many subsequent discussions about de-quantisation, as well as the anonymous referees for suggestions which improved paper.

\bibliographystyle{abbrv}

\end{document}